\documentclass[notitlepage,floatfix,aps,pra,superscriptaddress,twocolumn,footinbib]{revtex4-2}

\usepackage{graphicx}
\usepackage{bxtexlogo}
\usepackage{amsmath}
\usepackage{amsthm}
\usepackage{mathtools}
\usepackage{graphicx}
\usepackage{hyperref}
\usepackage{braket}
\usepackage{bm}
\usepackage{bbm}
\usepackage{ascmac}
\usepackage{mathrsfs}
\usepackage{xcolor}
\usepackage{comment}
\usepackage{ulem}

\allowdisplaybreaks

\DeclareMathOperator{\Tr}{Tr}
\DeclareMathOperator{\rank}{rank}

\newcommand{\sterling}[2]{\genfrac{[}{]}{0pt}{}{#1}{#2}} %

\newtheorem{dfn}{Definition}
\newtheorem{prop}{Proposition}
\newtheorem{lem}{Lemma}
\newtheorem{thm}{Theorem}
\newtheorem{cor}{Corollary}
\newtheorem{rem}{Remark}

\begin{document}

\title{Port-based telecloning of an unknown quantum state}
\author{Reiji Okada}
\email{reiji.okada@nagoya-u.jp}
\affiliation{Graduate School of Informatics, Nagoya University, Furo-cho, Chikusa-ku, Nagoya 464-8601, Japan}

\author{Kohtaro Kato}
\email{kokato@i.nagoya-u.ac.jp}
\affiliation{Graduate School of Informatics, Nagoya University, Furo-cho, Chikusa-ku, Nagoya 464-8601, Japan}

\author{Francesco Buscemi}
\email{buscemi@nagoya-u.jp}
\affiliation{Graduate School of Informatics, Nagoya University, Furo-cho, Chikusa-ku, Nagoya 464-8601, Japan}

\begin{abstract}
Telecloning is a protocol introduced by Murao \textit{et al.} [Phys. Rev. A \textbf{59}, 156 (1999)] to distribute copies of an unknown quantum state to many receivers in a way that beats the trivial ``clone-and-teleport'' protocol. In the last decade, a new type of teleportation called \textit{port-based teleportation}, in which the receiver can recover the state simply by looking at the correct port without having to actively perform correction operations, has been widely studied. In this paper, we consider the analog of telecloning, in which conventional teleportation is replaced by the port-based variant. To achieve this, we generalize the optimal measurement used in port-based teleportation and develop a new one that achieves port-based telecloning. Numerical results show that, in certain cases, the proposed protocol is strictly better than the trivial clone-and-teleport approach.
\end{abstract}

\maketitle

\section{Introduction}
Quantum teleportation~\cite{bennett1993teleporting} is one of the basic protocols in quantum communication, allowing the transmission of quantum information from one location to another without physically moving the particles carrying the quantum state and instead using only local operations and classical communications with preshared quantum entanglement. Quantum teleportation is used in quantum repeaters~\cite{briegel1998quantum} and is essential for the realization of long-distance quantum communication. The standard version of teleportation (ST) requires the receiver to actively perform a unitary correction on its system, depending on the classical information received from the sender.  

Port-based teleportation (PBT) is an alternative type of quantum teleportation proposed by Ishizaka and Hiroshima~\cite{ishizaka2008asymptotic, ishizaka2009quantum} that uses a multipartite entangled state whose subsystems are called ports. Unlike ST, PBT does not require the receiver to actively perform a unitary transformation; instead, the teleportation process is completed simply by selecting one of the multiple ports depending on the sender's measurement result and discarding the others. This feature enables PBT to be applied to universal programmable quantum processors~\cite{ishizaka2008asymptotic}. However, the no-programming theorem~\cite{nielsen1997programmable} shows that faithful and deterministic universal programmable quantum processors cannot be realized in a finite-dimensional system. Therefore, PBT succeeds only approximately or probabilistically. Despite these limitations, PBT has also been applied in areas such as instantaneous nonlocal quantum computation~\cite{beigi2011simplified} and communication complexity and Bell nonlocality~\cite{buhrman2016quantum}. There has been extensive research on the performance~\cite{wang2016higher, studzinski2017port, mozrzymas2018optimal, christandl2021asymptotic, leditzky2022optimality} and algorithms~\cite{fei2023efficient, grinko2023gelfand, grinko2023efficient,wills2024efficient} of PBT.

Telecloning, proposed by Murao \textit{et al.}~\cite{murao1999quantum, murao2000quantum}, is a protocol that generalizes teleportation with the goal of distributing a single unknown input state to many distant receivers. Since perfect copying of an unknown quantum state is forbidden by the no-cloning theorem~\cite{Park-no-cloning,wootters1982single,DIEKS1982271}, telecloning aims to transfer optimal clones instead~\cite{Buzek-Hillery-cloning,werner1998optimal}.

Existing telecloning protocols are based on ST and thus require receivers to actively perform unitary transformations to complete the protocol. In this paper, we introduce port-based telecloning (PBTC), which combines telecloning with (multi-)PBT and allows the transmission of copies of an unknown quantum state without requiring active corrections by the receivers. To this end, we generalize the measurement used in PBT and propose a new one that, when used on maximally entangled resource states, asymptotically achieves the fidelity of the optimal universal cloning protocol~\cite{werner1998optimal}. Interestingly, numerical results show that when the number of ports is small, PBTC achieves a transmission fidelity that is \textit{strictly higher} than that achievable by the naive method of simply performing optimal cloning and PBT sequentially. 
Throughout the study, we consider deterministic protocols using maximally entangled resource states. Room to investigate optimal resource states and a probabilistic variant of PBTC remains.  See Sec.~\ref{subsec:Port-based teleportation} for details.

The structure of this paper is as follows. In Sec.~\ref{sec:Preliminaries}, we summarize the necessary concepts of PBT and telecloning. In Sec.~\ref{sec:Port-based telecloning}, we introduce PBTC and explain its protocol. After discussing the generalization of the positive operator-valued measure (POVM) and its asymptotic optimality, we compare the performance with the trivial protocol. In Sec.~\ref{sec:Conclusion}, we provide a summary and discuss open questions.

\section{Preliminaries}\label{sec:Preliminaries}
\subsection{Port-based teleportation}\label{subsec:Port-based teleportation}

\begin{figure}[t]
    \centering
    \begin{minipage}{0.49\columnwidth}
        \centering
        \includegraphics[width=\linewidth, trim=0mm 65mm 130mm 0mm, clip]{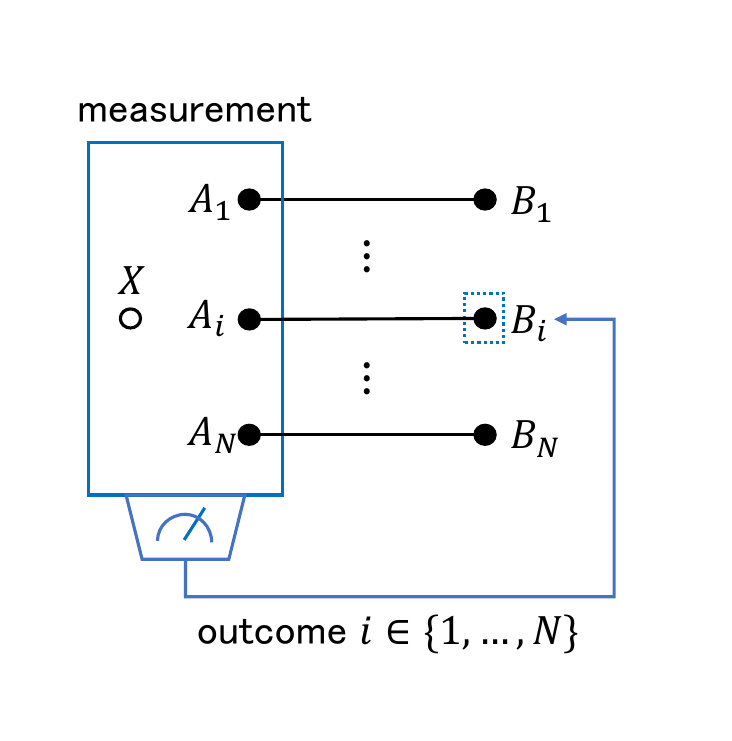}
        \smallskip
        (a) port-based teleportation
        \label{fig:PBT}
    \end{minipage}
    \hfill
    \begin{minipage}{0.49\columnwidth}
        \centering
        \includegraphics[width=\linewidth, trim=0mm 65mm 130mm 0mm, clip]{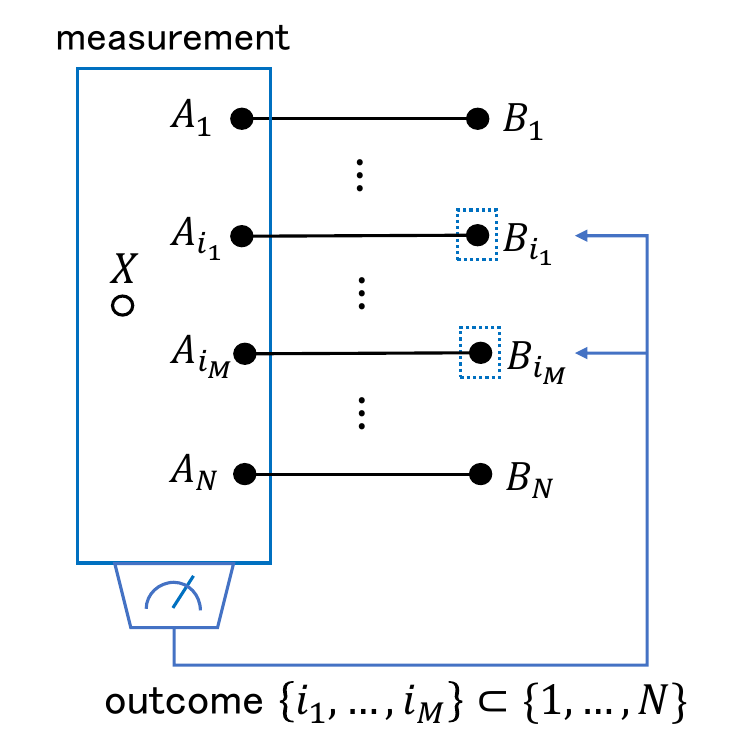}
        \smallskip
        (b) port-based telecloning
        \label{fig:PBTC}
    \end{minipage}
    \caption{(a) The setting of port-based teleportation (PBT). Alice and Bob share entangled states in ports $A_1,\ldots,A_N$ and $B_1,\ldots,B_N$, with Alice also holding the input state in the system $X$. In the protocol, Alice first measures systems $XA_1 \cdots A_N$ and sends the outcome $i\in\{1,\ldots,N\}$ to Bob, who then selects port $B_i$. This completes the protocol. (b) The setting of port-based telecloning (PBTC). PBTC is best thought of as having $N$ receivers, where the $i$th receiver has port $B_i\ (i=1,\ldots,N)$, and the goal of the protocol is to have $M$ receivers get approximate copies of the state of system $X$. In the protocol, Alice performs a measurement with a multi-index outcome, i.e., a subset $\{i_1,\ldots,i_M\}\subset\{1,\ldots,N\}$, and sends the outcome to all receivers. Each receiver keeps their port if their port number is contained in Alice's outcome. This completes the protocol, and an approximate copy of the input state of $X$ is transferred to $M$ receivers.}
    \label{fig:PBT and PBTC}
\end{figure}
For a finite-dimensional Hilbert space $\mathcal{H}$, $\mathcal{L}(\mathcal{H})$ denotes the space of linear operators on $\mathcal{H}$. 
In PBT, Alice (the sender) and Bob (the receiver) share an entangled state $\Phi_{A^NB^N}\in\mathcal{L}\big(\mathcal{H}^{\otimes N}\otimes\mathcal{H}^{\otimes N}\big)$ on $2N$ qudits, and Alice has the input pure state $\sigma_X\in\mathcal{L}(\mathcal{H}_X)$. Here, system $A^N\equiv A_1\cdots A_N\ (B^N\equiv B_1\cdots B_N)$ is on Alice's (Bob's) side. Each system $A$ is associated with a finite-dimensional Hilbert space $\mathcal{H}_A\cong\mathbbm{C}^d$. The shared entangled state $\Phi_{A^NB^N}$ is called a \textit{resource state}, and each of systems $A_1,\ldots,A_N$ and $B_1,\ldots,B_N$ is referred to as a \textit{port}. 

PBT is mainly classified into \textit{deterministic} PBT (dPBT) and \textit{probabilistic} PBT (pPBT). Although dPBT always succeeds, faithful transfer can be achieved only asymptotically. On the other hand, pPBT always achieves faithful transmission, but the unit success probability of the protocol is only asymptotically achieved. In this work, we focus on only the deterministic version, so we refer to dPBT simply as PBT. The protocol of PBT is as follows:
\begin{itemize}
    \item[(1)]Alice jointly measures the input system $X$ and all her ports $A^N$ with a POVM $\{E^i_{XA^N}\}_{i=1}^N$.
    \item[(2)]Alice relays the outcome $i\in\{1,\ldots,N\}$ to Bob via classical communication.
    \item[(3)]Bob selects the port $B_i$ and discards all other ports $B^N\setminus B_i$.
\end{itemize}
This completes the protocol, and the state is transferred to the remaining port $B_i$. The PBT channel $\mathcal{E}_N:\mathcal{L}(\mathcal{H}_X)\rightarrow\mathcal{L}(\mathcal{H}_B)$ is then expressed as follows:
\begin{align}
    \mathcal{E}_N(\sigma_X)=\sum_{i=1}^N\Tr_{XA^{N}B_i^c}\left[E^i_{XA^{N}}\left(\sigma_X\otimes\Phi_{A^NB^N}\right)\right],
\end{align}
where $B_i^c\equiv B^N\setminus B_i=B_1\cdots B_{i-1}B_{i+1}\cdots B_N$ and the remaining system $B_i$ is relabeled as output system $B$. The left side of Fig.~\ref{fig:PBT and PBTC} represents PBT~\cite{ishizaka2008asymptotic}.

The performance of PBT is evaluated by entanglement fidelity. The entanglement fidelity $F$ of a quantum channel $\mathcal{N}:\mathcal{L}(\mathcal{H})\rightarrow\mathcal{L}(\mathcal{\mathcal{H}})$ is defined as follows:
\begin{align}
  F(\mathcal{N})\coloneq\bra{\Phi^+}\left(\mathcal{N}\otimes \mathrm{id}\right)\left(\Phi^+\right)\ket{\Phi^+},
\end{align}
where $\Phi^+=\ket{\Phi^+}\bra{\Phi^+}$ is the maximally entangled state, defined by $\ket{\Phi^+}=\frac{1}{\sqrt{d}}\sum_{i=1}^d\ket{i}\ket{i}$ for the orthonormal basis $\{\ket{i}\}_{i=1}^d$, and $\mathrm{id}$ is the identity channel. Entanglement fidelity is related to average output fidelity $f$, which is defined as follows:
  \begin{align} f(\mathcal{N})\coloneq\int\mathrm{d}\phi\bra{\phi}\mathcal{N}(\phi)\ket{\phi},
\end{align}
where the integral is performed with respect to the uniform distribution $\mathrm{d}\phi$ over all input pure state. These two quantities are connected by the following relationship~\cite{horodecki1999general}:
\begin{align}\label{eq:f and F}
  f(\mathcal{N})=\frac{F(\mathcal{N})d+1}{d+1}.
\end{align}

An important class of POVMs in PBT is the pretty good measurement (PGM; also known as a square-root measurement)~\cite{belavkin1975optimal, hausladen1994pretty}. The PGM $\{E^i\}_i$ for the state ensemble $\{(p_i,\sigma^i)\}_{i\in\mathcal{I}}$ is given by
\begin{align}\label{eq:PGM}
  E^i=\bar{\sigma}^{-\frac{1}{2}}p_i\sigma^i\bar{\sigma}^{-\frac{1}{2}},
\end{align}
where $\bar{\sigma}=\sum_{i\in\mathcal{I}}p_i\sigma^i$ and $\bar{\sigma}^{-1}$ is defined on the support of $\bar{\sigma}$, which can always be assumed to be invertible, without loss of generality. Note that the sum of $E^i$ defined by \eqref{eq:PGM} is the projection onto the support of $\bar{\sigma}$, which generally does not coincide with the identity operator. One way to make them POVMs in the full Hilbert space is to add
\begin{align}
\label{eq:delta}
\Delta=\frac{1}{|\mathcal{I}|}\left(\mathbbm{1}-\sum_{i\in\mathcal{I}}E^i\right)
\end{align}
to each $E^i$. This does not change the argument in PBT  since $\Tr[\Delta\sigma^i]=0$.

References~\cite{ishizaka2009quantum, leditzky2022optimality} showed that the POVM that maximizes the fidelity of PBT is the PGM constructed from the ensemble $\{(1/N,\ \rho^i_{XA^N})\}_{i=1}^N$, where
\begin{align}\label{eq:signal state of PBT}
  \rho^i_{XA^N}\coloneq\Phi^+_{XA_i}\otimes\frac{1}{d^{N-1}}\mathbbm{1}_{A_i^c}\;.
\end{align}

A PBT protocol that uses $N$ pairs of maximally entangled states as ports and the PGM for $\{(1/N,\ \rho^i_{XA^N})\}_{i=1}^N$ as a measurement is called \textit{standard} PBT. The entanglement fidelity of the standard PBT channel $\mathcal{E}^\mathrm{std}_N$ is computed as follows~\cite{christandl2021asymptotic}:
\begin{align}
  F(\mathcal{E}^\mathrm{std}_N)=1-\frac{d^2-1}{4N}+O\left(N^{-\frac{3}{2}+\delta}\right),
\end{align}
where $\delta>0$. In addition, we can consider using any resource state and are not limited to maximally entangled states. Even in this case, it is known that the PGM that is the same as the POVM used in standard PBT is optimal (i.e., maximizing entanglement fidelity)~\cite{leditzky2022optimality}. If we denote a PBT channel using the PGM and optimal resource states as $\mathcal{E}^\mathrm{opt}_N$, the performance of such a channel can be expressed as follows~\cite{christandl2021asymptotic}:
\begin{align}
    F(\mathcal{E}^\mathrm{opt}_N)=1-\Theta(N^{-2}).
\end{align}

\subsection{Telecloning}\label{subsec:Telecloning}
Telecloning~\cite{murao1999quantum, murao2000quantum} is a generalization of ST to the case of many receivers. The objective of telecloning is to distribute one input state to many distant receivers. However, the no-cloning theorem~\cite{wootters1982single} prohibits making multiple perfect copies of a single unknown state. The best a sender can do is to transfer an \textit{optimal clone} that is the closest to the original state allowed by quantum mechanics. In the following, we will focus on symmetric cloning, i.e., the situation where there is no difference between the copies that each recipient receives.

The $K\rightarrow M$ optimal cloning map $\mathcal{C}:\mathcal{L}(\mathcal{H}^{\otimes K})\rightarrow\mathcal{L}(\mathcal{H}^{\otimes M})$ given by Werner~\cite{werner1998optimal} is obtained by projecting $K$ input copies and $M-K$ completely mixed states onto a symmetric subspace:
\begin{align}\label{eq:optimal cloning map}
  \mathcal{C}\left(\ket{\phi}\bra{\phi}^{\otimes K}\right)=\frac{d[K]}{d[M]}\Pi_M\left(\ket{\phi}\bra{\phi}^{\otimes K}\otimes\mathbbm{1}^{\otimes M-K}\right)\Pi_M,
\end{align}
where $d[K]=\binom{d+K-1}{K}$ and $\Pi_M$ is the projection onto the totally symmetric subspace of  $\mathcal{H}^{\otimes M}$. The state in~\eqref{eq:optimal cloning map} also optimizes the fidelity of each clone~\cite{keyl1999optimal}, which is written as follows:
\begin{align}
    \mathcal{R}\circ\mathcal{C}\left(\ket{\phi}\bra{\phi}^{\otimes K}\right)=\gamma_{K,M}\ket{\phi}\bra{\phi}+(1-\gamma_{K,M})\frac{1}{d}\mathbbm{1},
\end{align}
where $\gamma_{K,M}=\frac{K}{M}\frac{M+d}{K+d}$ and $\mathcal{R}$ represents the trace over all subsystems except the first one (due to exchange symmetry all clones are equal). Furthermore, since $\mathcal{C}$ is \textit{universal} (i.e., the fidelity does not depend on the input pure states), the fidelity for $K=1$ is given as follows:
\begin{align}
    f(\mathcal{R}\circ\mathcal{C})&=\gamma_{1,M}+(1-\gamma_{1,M})\frac{1}{d}\nonumber
    \\&=\frac{d+2M-1}{M(d+1)}.\label{eq:optimal fidelity}
\end{align}

There is a straightforward protocol to transfer optimal clones to many receivers. That is, Alice applies the optimal cloning map locally and transfers its output to each receiver by ST. If the number of clones is $M$, this protocol requires $M\log_2d$ ebits. Unlike this ``clone-and-teleport'' protocol, the protocol introduced in~\cite{murao1999quantum,murao2000quantum} performs cloning and teleportation simultaneously. An advantage of this protocol is that it requires only $(\log_2d)$-ebit entanglement between the sender and the receivers. This is achieved by using the $2M$-qudit entangled state, called a \textit{telecloning state}, which is shared between Alice and the $2M-1$ receivers (each participant has one qudit). The protocol is as follows:
\begin{enumerate}
    \item[(1)]Alice performs a complete $d$-dimensional Bell measurement on the input state and her entangled state.
    \item[(2)]Alice relays the outcome to all receivers via classical communication.
    \item[(3)]Each receiver applies an appropriate unitary transformation based on Alice's measurement outcome.
\end{enumerate}
This completes the protocol, and $M$ receivers obtain the optimal clone of the input state. This is possible because the universal cloning is covariant under the action of the unitary group.

\section{Port-based telecloning}\label{sec:Port-based telecloning}
\subsection{The protocol}
In this section, we introduce PBTC, which performs telecloning using PBT. The goal of PBTC is to distribute $M$ copies of the state of the input system across $N$ ports in one go. In particular, we consider a symmetric cloning scenario, i.e., all copies should look the same locally.

In PBTC, we use a POVM whose outcomes specify a subset of all ports available: The ports contained in such a subset will receive a copy of the input state, whereas the remaining ports will be discarded. The set of measurement outcomes is defined as follows.
\begin{dfn}
For fixed $N$ and $M\ (N\geq M)$, we define the set
\begin{align}
    \mathcal{I}_N^M\coloneq\{\{i_1,\ldots,i_M\}\ |\ &i_k\in\{1,\ldots,N\}\ \mathrm{for}\ k\in\{1,\ldots,M\}, \notag
    \\&\mathrm{and}\  i_1<\cdots<i_M\}.
\end{align}
\end{dfn}
\noindent Here, $|\mathcal{I}_N^M| = \binom{N}{M}$. For $I=\{i_1,\ldots,i_M\} \in \mathcal{I}_N^M$, we write the composite system $A_I \equiv A_{i_1} \cdots A_{i_M}$ and $A_I^c \equiv A^N \setminus A_I$.

The right side of Fig.~\ref{fig:PBT and PBTC} represents PBTC. In PBTC, Alice and $N$ receivers share a resource state $\Phi_{A^N B^N}$. We consider Alice to have ports $A_1,\ldots,A_N$, and the $i$th receiver has port $B_i\ (i=1,\ldots,N)$. Additionally, Alice holds the input state $\sigma_X$. The protocol of PBTC is as follows:
\begin{itemize}
    \item[(1)]Alice measures the input system $X$ and all her ports $A^N$ with a POVM $\{E^I_{X A^N}\}_{I \in \mathcal{I}_N^M}$.
    \item[(2)]Alice relays the outcome $I \in \mathcal{I}_N^M$ to all receivers via classical communication.
    \item[(3)]The $i$th receiver discards their port $B_i$ if $i \notin I$ and does nothing if $i \in I$ $(i=1,\ldots,N)$.
\end{itemize}
This completes the protocol, and the clones are transferred to the $M$ receivers. The PBTC channel $\mathcal{D}_{N,M}: \mathcal{L}(\mathcal{H}_X) \rightarrow \mathcal{L}(\mathcal{H}^{\otimes M})$ is expressed as follows:
\begin{align}
    \mathcal{D}_{N,M}(\sigma_X)=\sum_{I\in\mathcal{I}_N^M}\Tr_{XA^{N}B_I^c}\left[E^I_{XA^{N}}\left(\sigma_X\otimes\Phi_{A^NB^N}\right)\right],
\end{align} 
and the remaining system $B_{i_k}$ $(k=1, \ldots, M)$ is relabeled as output system $B_k$.

\subsubsection*{Clone-and-teleport protocol}
In Sec.~\ref{subsec:Telecloning}, we described a trivial ``clone-and-teleport'' protocol for telecloning, and we can consider a similar protocol for PBTC. The protocol is that Alice creates an optimal $M$ clone
locally and transfers it by \textit{multi-port-based teleportation} (MPBT)~\cite{studzinski2022efficient}. For simplicity, we consider only $1\to M$ cloning.

MPBT is the protocol that transfers $M$ qudit states to $M$ ports in one go. The POVM for MPBT using $N$ pairs of maximally entangled states is given by the PGM for $\{(|\mathcal{J}_N^M|^{-1},\rho^J_{X^MA^N})\}_{J\in\mathcal{J}_N^M}$, where
\begin{align}
    \mathcal{J}_N^M\coloneq\{(j_1,\ldots,j_M)\ |\ &j_k\in\{1,\ldots,N\}\ \mathrm{for}\ k\in\{1,\ldots,M\}, \notag
    \\&\mathrm{and}\  j_k\neq j_l\ \mathrm{for}\ k\neq l\}
\end{align}
is the ordered version of $\mathcal{I}_N^M$, and for $J=(j_1,\ldots,j_M)\in\mathcal{J}_N^M$,
\begin{align}
\rho^J_{X^MA^N}\coloneq\bigotimes_{k=1}^M\Phi^+_{X_kA_{j_k}}\otimes\frac{1}{d^{N-M}}\mathbbm{1}_{A_J^c}.
\end{align}

We refer to the protocol that performs optimal cloning and MPBT successively as the \textit{clone-and-MPBT protocol}. The clone-and-MPBT protocol can be considered in the framework of PBTC. The clone-and-MPBT protocol is equivalent to PBTC using the POVM
\begin{align}\label{eq:trivial POVM}
\big\{(\mathcal{C}^\dagger_{X^M\rightarrow X}\otimes\mathrm{id}_{A^N})(E^J_{X^MA^N})\big\}_{J\in\mathcal{J}_N^M},
\end{align}
where $\mathcal{C}^\dagger_{X^M\rightarrow X}:\mathcal{L}(\mathcal{H}^{\otimes M}_{X^M})\rightarrow\mathcal{L}(\mathcal{H}_X)$ is the adjoint of the optimal cloning map given by \eqref{eq:optimal cloning map} and $\{E^J_{X^MA^N}\}_{J\in\mathcal{J}_N^M}$ is the PGM for $\{(|\mathcal{J}_N^M|^{-1},\rho^J_{X^MA^N})\}_{J\in\mathcal{J}_N^M}$. The set of measurement outcomes in \eqref{eq:trivial POVM} is not $\mathcal{I}_N^M$, but that poses no issue because it can be made equivalent to $\mathcal{I}_N^M$ by summing the POVM elements for outcomes that are identical when reordered. Note that since the figure of merit for \eqref{eq:trivial POVM} differs from the original PGM, it is essential to make \eqref{eq:trivial POVM} a proper POVM in the full Hilbert space. Thus, we take into account $\Delta$, given by \eqref{eq:delta}, in the clone-and-MPBT protocol.

The clone-and-MPBT protocol can transfer optimal clones in the limit of the number of ports $N\rightarrow\infty$. However, optimality for finite $N$ is not guaranteed. In fact, the POVM we introduce in the next section achieves higher fidelity than the clone-and-MPBT protocol when $N$ is small.

\subsection{Generalization of POVM}
As we have noted, in this work, we consider only symmetric cloning. We first introduce an ensemble  for a PGM by partially symmetrizing the state $\rho^i_{XA^N}$ that constitutes the optimal POVM of PBT.
\begin{dfn}
For $I=\{i_1,\ldots,i_M\}\in\mathcal{I}_N^M$, let
\begin{align}\label{eq:signal state of PBTC}
  \eta^{I}_{XA^N}\coloneq\frac{d^M}{d[M]}\Pi_{A_I}\rho^{i_1}_{XA^N}\Pi_{A_I},
\end{align}
where $\rho^{i_1}_{XA^N}$ is the state given by \eqref{eq:signal state of PBT} and $\Pi_{A_I}$ is the projection onto the symmetric subspace of $\mathcal{H}^{\otimes M}_{A_I}$. 
\end{dfn}
\noindent In Eq. \eqref{eq:signal state of PBTC}, although we formally use $i_1$ as the index of $\rho^{i_1}_{XA^N}$, note that $\eta^I_{XA^N}$ remains in the same state regardless of whether $i_1$ is replaced by any element of $I=\{i_1,\ldots,i_M\}$. 

We refer to PBTC that uses $N$ pairs of maximally entangled states and the PGM for $\{(|\mathcal{I}_N^M|^{-1},\eta^I_{XA^N})\}_{I\in\mathcal{I}_N^M}$ as \textit{standard} PBTC. Since standard PBTC is symmetric cloning, we evaluate its performance using an average fidelity of a single clone for all input pure states.

The asymptotic fidelity of standard PBTC is given by the following theorem.
\begin{thm}\label{thm:main result}
  Let us consider the standard PBTC channel~$\mathcal{D}^\mathrm{std}_{N,M}$ and the channel $\mathcal{R}$ that represents the trace over all subsystems except the first one. In the limit of the number of ports $N\rightarrow\infty$, the following equality holds:
    \begin{align}\label{eq:fidelity of standard PBTC}
        \lim_{N\to \infty}f(\mathcal{R}\circ\mathcal{D}^\mathrm{std}_{N,M})=\frac{d+2M-1}{M(d+1)},
    \end{align}
    where $d$ is the dimension of the local Hilbert space and $M$ is the number of clones to be transferred.
\end{thm}
\noindent The proof is given in Sec.~\ref{subsec:proof}. The value of \eqref{eq:fidelity of standard PBTC} coincides with the fidelity of $1\rightarrow M$ optimal cloning given by \eqref{eq:optimal fidelity}. Therefore, standard PBTC can transfer optimal clones asymptotically.

Finally, we numerically compare the performance of standard PBTC with the clone-and-MPBT protocol described in the previous section.
Figure~\ref{fig:Fidelity} shows the fidelity of each protocol obtained by numerical calculation. Fidelity is calculated for a single clone. Namely, it represents $f(\mathcal{R}\circ\mathcal{D}^\mathrm{std}_{N,M})$ and $f(\mathcal{R}\circ\mathcal{T}_{N,M})$, where $\mathcal{D}^\mathrm{std}_{N,M}$ is the standard PBTC channel and $\mathcal{T}_{N,M}$ is the quantum channel corresponding to the clone-and-MPBT protocol. 
Figure~\ref{fig:Fidelity} shows that standard PBTC achieves higher fidelity than the clone-and-MPBT protocol when $\nobreak{d=2},\ \nobreak{M=2}$, and $\nobreak{2\leq N\leq6}$. Due to the increasing complexity of the calculation we were not able to go to higher values of $M$ and $N$, but we conjecture that a finite gap exists for all finite values. This problem could potentially be gotten rid of by exploiting the proper representation theory approach given in \cite{studzinski2022efficient,grinko2023gelfand}.

\begin{figure}[t]
\includegraphics[width=\columnwidth, trim=0mm 90mm 0mm 90mm, clip]{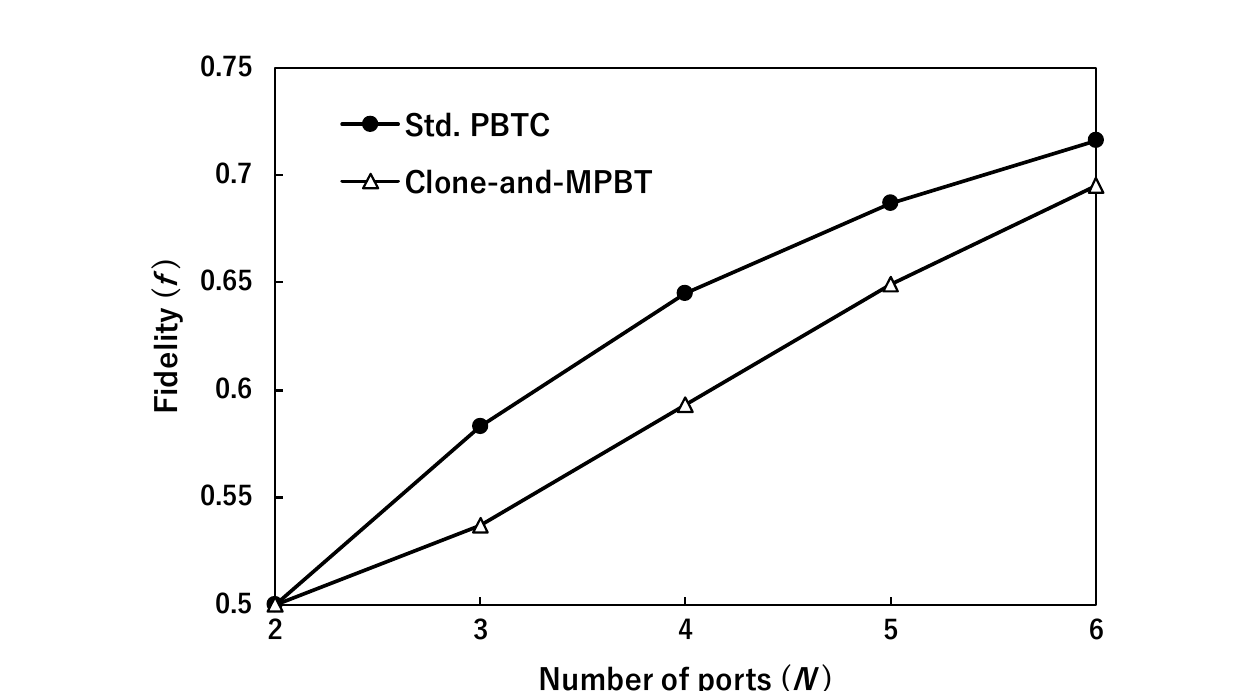}
\caption{The fidelity of the PBTC protocol proposed here (circles) and the trivial clone-and-MPBT protocol (triangles). The plotted values were obtained by numerical calculation for $d=2$ and $M=2$. Fidelity is the average over the input pure state and is calculated for a single clone.}
\label{fig:Fidelity}
\end{figure}

\subsection{Proof of Theorem~\ref{thm:main result}}\label{subsec:proof}
In this section, we prove Theorem \ref{thm:main result}. Within this section, we use the same notation for operators in systems $A^NB$ as we did for operators in systems $XA^N$ in the previous sections via the isomorphism $X\cong B$. 
For example, for $\rho^i_{XA^N}=\Phi^+_{XA_i}\otimes\frac{1}{d^{N-1}}\mathbbm{1}_{A_i^c}$ defined in \eqref{eq:signal state of PBT}, we have $\rho^i_{A^NB}=\Phi^+_{A_iB}\otimes\frac{1}{d^{N-1}}\mathbbm{1}_{A_i^c}$.

We start by showing the properties related to the symmetric group.

\begin{dfn}\label{dfn:permutation}
Let $S_N$ be the symmetric group in $\{1,\ldots,N\}$, and for $I\in\mathcal{I}_N^M$, let $S_I$ be the subgroup of $S_N$ consisting of all permutations of $I=\{i_1,\ldots,i_M\}$. For $\sigma\in S_N$, the action of the unitary representation $V_\sigma\in\mathcal{L}(\mathcal{H}^{\otimes N})$ is defined follows:
 \begin{align}
     V_\sigma\ket{k_1\cdots k_N}\coloneq\ket{k_{\sigma^{-1}(1)}\cdots k_{\sigma^{-1}(N)}}.
 \end{align}
In addition, for $I\in\mathcal{I}_N^M$, let $\Pi_{A_I}$ be the projection onto the symmetric subspace of $\mathcal{H}^{\otimes M}_{A_I}$ as follows: 
\begin{align}\label{eq:projection onto the symmetric subspace}
    \Pi_{A_I}\coloneq\frac{1}{M!}\sum_{\sigma\in S_I}V_\sigma.
\end{align}
\end{dfn} 

\noindent Note that although $\Pi_{A_I}$ defined in \eqref{eq:projection onto the symmetric subspace} is an operator on $\mathcal{H}^{\otimes N}_{A^N}$, it acts nontrivially only on $\mathcal{H}^{\otimes M}_{A_I}$.

\begin{lem}\label{lem:sigma and S}
   For any $\sigma\in S_N$ and $I\in\mathcal{I}_N^M$, $\sigma S_I\sigma^{-1}=~S_{\sigma(I)}$.  
\end{lem}

\begin{proof}
First, if $\tau\in\sigma S_I\sigma^{-1}$, a $\pi\in S_I$ such that $\tau=\sigma\pi\sigma^{-1}$ exists. Since $\sigma\pi\sigma^{-1}$ is a bijection on $\sigma(I)$, $\tau$ is a permutation of $\sigma(I)$. Thus, $\sigma S_I\sigma^{-1}\subset S_{\sigma(I)}$. Next, suppose $\tau\in S_{\sigma(I)}$. In this case, $\sigma^{-1}\tau\sigma$ is a bijection on $I$. Therefore, a $\pi\in S_I$ such that $\sigma^{-1}\tau\sigma=\pi$ exists. Since $\tau=\sigma\pi\sigma^{-1}$ and $\sigma\pi\sigma^{-1}\in\sigma S_I\sigma^{-1}$, we have $\tau\in\sigma S_I\sigma^{-1}$. Therefore, $S_{\sigma(I)}\subset\sigma S_I\sigma^{-1}$.
\end{proof}

\begin{cor}\label{cor:sigma and P}
For any $\sigma\in S_N$ and $I\in\mathcal{I}_N^M$, $V_\sigma \Pi_{A_I}V_\sigma^\dagger=\Pi_{A_{\sigma(I)}}$.
\end{cor}

\begin{proof}
  From Lemma \ref{lem:sigma and S}, we have
  \begin{align}
    \sum_{\tau\in S_I}V_{\sigma\tau\sigma^{-1}}=\sum_{\pi\in S_{\sigma(I)}}V_{\pi}.
  \end{align}
  Since $V_{\sigma\pi}=V_\sigma V_\pi$ and $V_{\sigma^{-1}}=V_\sigma^\dagger$ for any $\sigma,\pi\in S_N$, the lemma holds.
\end{proof}

\begin{prop}\label{prop:P invariance of E}
  Let $\left\{E^{I}_{A^NB}\right\}_{I\in\mathcal{I}_N^M}$ be the PGM for $\{(|\mathcal{I}_N^M|^{-1},\eta^I_{A^NB})\}_{I\in\mathcal{I}_N^M}$. For any $I=\{i_1,\ldots,i_M\}\in\mathcal{I}_N^M$, it holds that $\Pi_{A_I}E^I_{A^NB}\Pi_{A_I}=E^I_{A^NB}$.
\end{prop}

\begin{proof}
  Let us denote
  \begin{align}
    \bar{\eta}_{A^NB}=\binom{N}{M}^{-1}\sum_{I\in\mathcal{I}_N^M}\eta^I_{A^NB}.   
  \end{align}
 For any $\sigma\in S_N$, we have
  \begin{align}
    &V_\sigma\bar{\eta}_{A^NB}\nonumber
    \\=&{\binom{N}{M}}^{-1}\frac{d^M}{d[M]}\sum_{I\in\mathcal{I}_N^M}V_{\sigma}\Pi_{A_I}\rho^{i_1}_{A^NB}\Pi_{A_I}\nonumber
\\=&{\binom{N}{M}}^{-1}\frac{d^M}{d[M]}\sum_{I\in\mathcal{I}_N^M}V_{\sigma}\Pi_{A_I}V_{\sigma}^\dagger V_\sigma\rho^{i_1}_{A^NB}V_{\sigma}^\dagger V_\sigma \Pi_{A_I}V_{\sigma}^\dagger V_\sigma\nonumber
   \\=&{\binom{N}{M}}^{-1}\frac{d^M}{d[M]}\sum_{I\in\mathcal{I}_N^M}\Pi_{A_{\sigma(I)}}\rho^{\sigma(i_1)}_{A^NB}\Pi_{A_{\sigma(I)}}V_\sigma\nonumber
   \\=&\bar{\eta}_{A^NB}V_\sigma
  \end{align}
  The second equality uses $V_\sigma^{-1}=V_\sigma^\dagger$, and the third equality uses $V_\sigma\rho^{i_1}_{A^NB}V_{\sigma}^\dagger=\rho^{\sigma(i_1)}_{A^NB}$ and Corollary \ref{cor:sigma and P}. Note that from the symmetry ${\binom{N}{M}}^{-1}\frac{d^M}{d[M]}\sum_{I\in\mathcal{I}_N^M}\Pi_{A_{\sigma(I)}}\rho^{\sigma(i_1)}_{A^NB}\Pi_{A_{\sigma(I)}}={\bar \eta}_{A^NB}$ holds. By summing both sides with respect to $\sigma\in S_I$ and dividing by $M!$, we obtain
  \begin{align}
    \Pi_{A_I}\bar{\eta}_{A^NB}=\bar{\eta}_{A^NB}\Pi_{A_I}.
  \end{align}
  Since $\bar{\eta}^{-1}_{A^NB}$ is defined on the support of $\bar{\eta}_{A^NB}$, $\big[\Pi_{A_I},\bar{\eta}^{-\frac{1}{2}}_{A^NB}\big]=0$ also holds. Therefore,
  \begin{align}
    &\Pi_{A_I}E^I_{A^NB}\Pi_{A_I}\nonumber
    \\=&\Pi_{A_I}\bar{\eta}_{A^NB}^{-\frac{1}{2}}\left(\frac{d^M}{d[M]}\Pi_{A_I}\rho^{i_1}_{A^NB}\Pi_{A_I}\right)\bar{\eta}_{A^NB}^{-\frac{1}{2}}\Pi_{A_I}\nonumber
    \\=&\bar{\eta}_{A^NB}^{-\frac{1}{2}}\left(\frac{d^M}{d[M]}\Pi_{A_I}\rho^{i_1}_{A^NB}\Pi_{A_I}\right)\bar{\eta}_{A^NB}^{-\frac{1}{2}}\nonumber
    \\=&E^I_{A^NB}.
  \end{align}
\end{proof}

We then calculate the entanglement fidelity. The following lemma connects the entanglement fidelity of
PBT with the state-discrimination problem.
\begin{lem}[\cite{ishizaka2009quantum,beigi2011simplified}]
\label{lem:entanglement fidelity of standard dPBT}
Let us fix the resource state to be $N$ pairs of maximally entangled states. The entanglement fidelity of the PBT channel $\mathcal{E}_N$ using the POVM $\{E^i_{XA^N}\}_{i=1}^N$ is given by
\begin{align}\label{eq:entanglement fidelity of standard dPBT}
  F(\mathcal{E}_N)=\frac{1}{d^2}\sum_{i=1}^N\Tr\left[E^i_{A^NB}\rho^i_{A^NB}\right].
\end{align}
\end{lem}

\noindent By applying this lemma for standard PBTC, we obtain the following corollary.
\begin{cor}\label{cor:F of PBTC}
Let us consider the standard PBTC channel~$\mathcal{D}^\mathrm{std}_{N,M}$ and quantum channel $\mathcal{R}$ that traces over all subsystems except the first one. The entanglement fidelity of the quantum channel $\mathcal{R}\circ\mathcal{D}^\mathrm{std}_{N,M}$ is given by
\begin{align} 
    F(\mathcal{R}\circ\mathcal{D}^\mathrm{std}_{N,M})=\frac{1}{d^2}\sum_{I\in\mathcal{I}_N^M}\Tr\left[E^I_{A^NB}\rho^{i_1}_{A^NB}\right],
\end{align}
where $i_1$ is the smallest number of $I$ and $\{E^I_{XA^N}\}_{I\in\mathcal{I}_N^M}$ is the PGM for $\{(|\mathcal{I}_N^M|^{-1},\eta^I_{A^NB})\}_{I\in\mathcal{I}_N^M}$.
\end{cor}
\noindent The following lemma provides the lower bound of the success probability for the state-discrimination problem of PGM. It is proportional to the entanglement fidelity, as expressed by Lemma \ref{lem:entanglement fidelity of standard dPBT}.
\begin{lem}[\cite{beigi2011simplified}]
\label{lem:general lower bound of PGM}
Let $\{E^i\}_{i=1}^N$ be the PGM for any state ensemble $\{(1/N, \sigma^i)\}_{i=1}^N$. Then, the success probability for the state-discrimination problem
\begin{align}
    p_\mathrm{succ}=\frac{1}{N}\sum_{i=1}^N\Tr\left[E^i\sigma^i\right]
\end{align}
satisfies the following inequality:
\begin{align}
     p_\mathrm{succ}\geq\frac{1}{N\bar{r}\Tr\bar{\sigma}^2},
\end{align}
where
\begin{align}\label{eq:parts}
    \bar{r}=\frac{1}{N}\sum_{i=1}^N\rank{\sigma^i}\ ,\ \bar{\sigma}=\frac{1}{N}\sum_{i=1}^N{\sigma^i}.
\end{align}
\end{lem}

To utilize Lemma~\ref{lem:general lower bound of PGM}, we calculate the values of \eqref{eq:parts} for the state ensemble $\{(|\mathcal{I}_N^M|^{-1},\eta^I_{XA^N})\}_{I\in\mathcal{I}_N^M}$. The average rank is given by the following proposition.
\begin{prop}\label{prop:bar r}
$\binom{N}{M}^{-1}\sum_{I\in\mathcal{I}_N^M}\rank{\eta^I_{A^NB}}=d[M-1]d^{N-M}.$
\end{prop}
\begin{proof}
For any $I=\{i_1,\ldots,i_M\}\in\mathcal{I}_N^M$, we have
  \begin{align}
\eta^I_{A^NB}=\frac{d^{M-N+1}}{d[M]}\Pi_{A_I}(\Phi^+_{A_{i_1}B}\otimes\mathbbm{1}_{A_I\setminus A_{i_1}})\Pi_{A_I}\otimes\mathbbm{1}_{A_I^c}.
  \end{align}
 The set $W$ of eigenvectors corresponding to the non-zero eigenvalues of $\Pi_{A_I}(\Phi^+_{A_{i_1}B}\otimes\mathbbm{1}_{A_I\setminus A_{i_1}})\Pi_{A_I}$ is
\begin{align}
W=\left\{\Pi_{A_I}\ket{\Phi^+}_{A_{i_1}B}\ket{k_1\cdots k_{M-1}}_{A_I\setminus A_{i_1}}\right\}_{k_1,\ldots,k_{M-1}=1}^d,
\end{align}
where $\{\ket{k}\}_{k=1}^d$ is the orthonormal basis of $\mathcal{H}$. Since $\rank\mathbbm{1}_{A_I^c}=d^{N-M}$, it is sufficient to show $|W|=d[M-1]$. Here,
\begin{align}\label{eq:eq in rank}
  &\Pi_{A_I}\ket{\Phi^+}_{A_{i_1}B}\ket{k_1\cdots k_{M-1}}_{A_I\setminus A_{i_1}}\nonumber
  \\=&\frac{1}{\sqrt{d}}\sum_{i=1}^d\Pi_{A_I}\ket{i\ k_1\cdots k_{M-1}}_{A_I}\ket{i}_B.
\end{align}
Thus, 
\begin{align}
  &\Pi_{A_I}\ket{\Phi^+}_{A_{i_1}B}\ket{k_1\cdots k_{M-1}}_{A_I\setminus A_{i_1}}\nonumber
  \\=&\Pi_{A_I}\ket{\Phi^+}_{A_{i_1}B}\ket{l_1\cdots l_{M-1}}_{A_I\setminus A_{i_1}}
\end{align}
holds for $k_1,\ldots,k_{M-1},l_1,\ldots,l_{M-1}\in\{1,\ldots d\}$ if and only if
\begin{align}\label{eq:constraint}
  \Pi_{A_I}\ket{i\ k_1\cdots k_{M-1}}_{A_I}=\Pi_{A_I}\ket{i\ l_1\cdots l_{M-1}}_{A_I}
\end{align}
holds for each $i\in\{1,\ldots,d\}$. Equation \eqref{eq:constraint} holds if and only if $k_1,\ldots,k_{M-1}$ and $l_1,\ldots,l_{M-1}$ match after permutation. Thus, $|W|$ equals the number of combinations with repetition of selecting $M-1$ elements from $\{1,\ldots,d\}$. Hence, $|W|=d[M-1]$.
\end{proof}

Next, we estimate $\Tr{\bar{\eta}^2_{A^NB}}$.
\begin{lem}\label{lem:Tr depends on multiplicity of alphabets}
  Let $I,J,K,L\in\mathcal{I}_N^M$. If $|I\cap J|=|K\cap L|$, then $\Tr\left[\eta^I_{A^NB}\eta^J_{A^NB}\right]=\Tr\left[\eta^K_{A^NB}\eta^L_{A^NB}\right]$.
\end{lem}
\begin{proof}
  For any $\sigma\in S_N$ and $I=\{i_1,\ldots,i_M\}\in\mathcal{I}_N^M$, we have
  \begin{align}
    \eta^I_{A^NB}&=\frac{d^M}{d[M]}\Pi_{A_I}\rho^{i_1}_{A^NB}\Pi_{A_I}\nonumber
    \\&=\frac{d^M}{d[M]}V_\sigma^\dagger V_\sigma \Pi_{A_I}V_\sigma^\dagger V_\sigma\rho^{i_1}_{A^NB}V_\sigma^\dagger V_\sigma \Pi_{A_I}V_\sigma^\dagger V_\sigma\nonumber
    \\&=\frac{d^M}{d[M]}V_\sigma^\dagger \Pi_{A_{\sigma(I)}}\rho^{\sigma(i_1)}_{A^NB}\Pi_{A_{\sigma(I)}}V_\sigma\nonumber
    \\&=V_\sigma^\dagger\eta^{\sigma(I)}_{A^NB}V_\sigma.
  \end{align}
  The second equality uses $V_\sigma^{-1}=V_\sigma^\dagger$, and the third equality uses $V_\sigma\rho^{i_1}_{A^NB}V_{\sigma}^\dagger=\rho^{\sigma(i_1)}_{A^NB}$ and Corollary \ref{cor:sigma and P}. Therefore, for any $\sigma\in S_N$ and $I,J\in\mathcal{I}_N^M$, we obtain
  \begin{align}
    \Tr\left[\eta^I_{A^NB}\eta^J_{A^NB}\right]&=\Tr\left[(V_\sigma^\dagger\eta^{\sigma(I)}_{A^NB}V_\sigma)(V_\sigma^\dagger\eta^{\sigma(J)}_{A^NB}V_\sigma)\right]\nonumber
    \\&=\Tr\left[\eta^{\sigma(I)}_{A^NB}\eta^{\sigma(J)}_{A^NB}\right].
  \end{align}
Note that $|I\cap J|=|\sigma(I)\cap\sigma(J)|$. Moreover, for any $K,L\in\mathcal{I}_N^M$ satisfying $|I\cap J|=|K\cap L|$, a $\sigma\in S_N$ such that $\sigma(I)=K$ and $\sigma(J)=L$ exists. Thus, the proposition is proved.
\end{proof}

\begin{lem}\label{lem:Tr takes maximum value when I=J}
  For any $I,J\in\mathcal{I}_N^M$, it holds that $\Tr\left[\eta^I_{A^NB}\eta^J_{A^NB}\right]\leq\Tr\left[(\eta^I_{A^NB})^2\right]$.
\end{lem}

\begin{proof}
  Applying the Cauchy-Schwarz inequality to the Hilbert-Schmidt inner product, we have
  \begin{align}
     \left|\Tr[A^\dagger B]\right|\leq\sqrt{\Tr\left[A^\dagger A\right]}\sqrt{\Tr\left[B^\dagger B\right]}.
  \end{align}
 for any $A,B\in\mathcal{L}(\mathcal{H}^{\otimes N+1})$. By setting $A=\eta^I_{A^NB}$ and $B=V_\sigma\eta^I_{A^NB}V_\sigma^\dagger=\eta^{\sigma(I)}_{A^NB}$ for $\sigma\in S_N$, the left-hand side can be written as
\begin{align}
  \left|\Tr\left[\eta^I_{A^NB}\eta^{\sigma(I)}_{A^NB}\right]\right|=\Tr\left[\eta^I_{A^NB}\eta^{\sigma(I)}_{A^NB}\right],
\end{align}
while the right-hand side becomes
\begin{align}
&\sqrt{\Tr\left[\eta^I_{A^NB}\eta^I_{A^NB}\right]}\sqrt{\Tr\left[(V_\sigma\eta^I_{A^NB}V_\sigma^\dagger)(V_\sigma\eta^I_{A^NB}V_\sigma^\dagger)\right]}\nonumber
  \\=&\Tr\left[(\eta^I_{A^NB})^2\right].
\end{align}
Since $\sigma\in S_N$ is arbitrary, $J=\sigma(I)$ is arbitrary.
\end{proof}

\begin{lem}\label{lem:upper bound of Tr when I=J}
For any $I\in\mathcal{I}_N^M$, the following inequality holds:
\begin{align}
    \Tr\left[(\eta^I_{A^NB})^2\right]\leq\frac{d^{M-N+2}}{d[M]}\frac{M!(M-1)!}{d+M-1}.
\end{align}
\end{lem}

\begin{proof}
From the arguments made in the proof of Proposition \ref{prop:bar r}, $\eta^I_{A^NB}$ has $d[M-1]d^{N-M}$ non-zero eigenvalues, and its eigenvectors are given in the form
\begin{align}\label{eq:eigen vector}
   \Pi_{A_I}\ket{\Phi^+}_{A_{i_1}B}\ket{k_1\cdots k_{M-1}}_{A_I\setminus A_{i_1}}\ket{k_M\cdots k_{N-1}}_{A_I^c},
\end{align}
where $k_1,\ldots,k_{N-1}\in\{1,\ldots,d\}$ and $\{\ket{k}\}_{k=1}^d$ is the orthonormal basis of $\mathcal{H}$. If there are $n$ different permutations to rearrange $k_1,\ldots, k_{M-1}$ without distinguishing the same numbers, the eigenvalue corresponding to \eqref{eq:eigen vector} is $n\frac{d^{M-N+1}}{d[M]}$. Since $n\leq(M-1)!$ always holds, all $d[M-1]d^{N-M}$ nonzero eigenvalues are less than or equal to $(M-1)!\frac{d^{M-N+1}}{d[M]}$. Hence,
\begin{align}
  \Tr\left[(\eta^I_{A^NB})^2\right]&\leq d[M-1]d^{N-M}\left((M-1)!\frac{d^{M-N+1}}{d[M]}\right)^2\nonumber
  \\&=\frac{d^{M-N+2}}{d[M]}\frac{M!(M-1)!}{d+M-1}.
\end{align}
\end{proof}

\begin{lem}\label{lem:value of Tr when I cap J = 0}
 Let $I,J\in\mathcal{I}_N^M$. If $I\cap J=\emptyset$, then $\Tr\big[\eta^I_{A^NB}\eta^J_{A^NB}\big]=1/d^{N+1}$.
\end{lem}
\begin{proof}
Suppose $I=\{i_1,\ldots,i_M\}$ and $J=\{j_1,\ldots,j_M\}$. When $I\cap J=\emptyset$, we have
\begin{align}
    \big[\Pi_{A_I},\Pi_{A_J}\big]=\big[\Pi_{A_I},\rho^{j_1}_{A^NB}\big]=\big[\Pi_{A_J},\rho^{i_1}_{A^NB}\big]=0.
\end{align} 
Thus,
\begin{align}\label{eq:Tr}
    \Tr\left[\eta^I_{A^NB}\eta^J_{A^NB}\right]=\frac{d^{2M}}{(d[M])^2}\Tr\left[\Pi_{A_I}\rho^{i_1}_{A^NB}\Pi_{A_J}\rho^{j_1}_{A^NB}\right].
\end{align}
Therefore, with some calculations we obtain
\begin{align}
      &\Tr\left[\eta^I_{A^NB}\eta^J_{A^NB}\right]\nonumber
      \\=&\frac{1}{(d[M]M!)^2d^N}\sum_{\sigma\in S_I}\sum_{\tau\in S_J}\sum_{k_{i_1},\ldots,k_{i_M}=1}^d\sum_{l_{j_1},\ldots,l_{j_M}=1}^d\nonumber
      \\&\times\delta_{l_{j_1},k_{\sigma^{-1}(i_1)}}\delta_{k_{i_1},l_{\tau^{-1}(j_1)}}\prod_{p=2}^M\delta_{k_{i_p},k_{\sigma^{-1}(i_p)}}\delta_{l_{j_p},l_{\tau^{-1}(j_p)}}.\label{eq:appendix}
\end{align}
For a detailed derivation of \eqref{eq:appendix}, see the Appendix.
To calculate \eqref{eq:appendix}, we decompose $\sigma$ and $\tau$ into cycles. Suppose $\sigma\in S_I$ can be decomposed into $\sigma=C_1^\sigma C_2^\sigma\cdots C_{l(\sigma)}^\sigma$ (including cycles with a single element), and let $c_m^\sigma$ represent the length of the cycle $C_m^\sigma$ (with similar notation for $\tau\in S_J$). Then, by definition, 
\begin{align}
    \sum_{m=1}^{l(\sigma)}c_m^\sigma=\sum_{m=1}^{l(\tau)}c_m^\tau=M.
\end{align}
Furthermore, we denote the elements of the cycle as $C_m^\sigma=(i^m_1\ \cdots\ i^m_{c^\sigma_i})$ and $C_m^\tau=(j^m_1\ \cdots\ j^m_{c^\tau_j})$. Without loss of generality, let $i_1=i_1^1$ and $j_1=j_1^1$. When $c^\sigma_1\neq 1$ and $c^\tau_1\neq 1$, each summand in \eqref{eq:appendix} can be expressed as follows:
\begin{align}\label{eq:after}
&\delta_{l_{j_1},k_{\sigma^{-1}(i_1)}}\delta_{k_{i_1},l_{\tau^{-1}(j_1)}}\prod_{p=2}^M\delta_{k_{i_p},k_{\sigma^{-1}(i_p)}}\delta_{l_{j_p},j_{\tau^{-1}(j_p)}}\nonumber
\\=&\left(\delta_{l_{j^1_1},k_{\sigma^{-1}\left(i^1_1\right)}}\delta_{k_{i^1_1},l_{\tau^{-1}\left(j^1_1\right)}}\prod_{p=2}^{c^\sigma_1}\delta_{k_{i^1_p},k_{\sigma^{-1}\left(i^1_p\right)}}\prod_{q=2}^{c^\tau_1}\delta_{l_{j^1_q},l_{\tau^{-1}\left(j^1_q\right)}}\right)\nonumber
\\&\times\left(\prod_{m=2}^{l(\sigma)}\prod_{n=1}^{c^\sigma_m}\delta_{k_{i^m_n},k_{\sigma^{-1}\left(i^m_n\right)}}\right)\left(\prod_{s=2}^{l(\tau)}\prod_{t=1}^{c^\tau_s}\delta_{l_{j^s_t},l_{\tau^{-1}(j^s_t)}}\right).
\end{align}
When $c^\sigma_1=1$ or $c^\tau_1=1$, $\prod_{p=2}^{c^\sigma_1}\delta_{k_{i^1_p},k_{\sigma^{-1}\left(i^1_p\right)}}$  or $\prod_{q=2}^{c^\tau_1}\delta_{l_{j^1_q},l_{\tau^{-1}\left(j^1_q\right)}}$  in \eqref{eq:after} is disregarded, respectively. The right-hand side of \eqref{eq:after} corresponds to cycles $C^\sigma_1$ and $C^\tau_1$ in the first parentheses, $C^\sigma_2,\ldots,C^\sigma_{l(\sigma)}$ in the second parentheses, and $C^\tau_2,\ldots,C^\tau_{l(\tau)}$ in the third parentheses. When we sum \eqref{eq:after} over $k_{i_1},\ldots,k_{i_M}$ and $l_{j_1},\ldots,l_{j_M}$, the first parentheses eliminate $c^\sigma_1+c^\tau_1-1$ indices. The second parentheses eliminate $c^\sigma_m-1$ indices for fixed $m\in\{2,\ldots,l(\sigma)\}$, for a total of $\sum_{m=2}^{l(\sigma)}(c^\sigma_m-1)$ indices. Similarly, the third parentheses eliminate $\sum_{s=2}^{l(\tau)}(c^\tau_s-1)$ indices. Consequently, the total number of eliminated indices for fixed $\sigma\in S_I,\tau\in S_J$ is
\begin{align}
&(c^\sigma_1+c^\tau_1-1)+\sum_{m=2}^{l(\sigma)}(c^\sigma_m-1)+\sum_{s=2}^{l(\tau)}(c^\tau_s-1)\nonumber
\\=&\sum_{k=1}^{l(\sigma)}c_k^\sigma+\sum_{k=1}^{l(\tau)}c_k^\tau-1-[l(\sigma)-1]-[l(\tau)-1]\nonumber
  \\=&2M+1-l(\sigma)-l(\tau).
\end{align}
Since there were $2M$ indices of $k_{i_1},\ldots,k_{i_M}$ and $l_{j_1},\ldots,l_{j_M}$ at the beginning, the remaining indices are
\begin{align}
  2M-\left[2M+1-l(\sigma)-l(\tau)\right]=l(\sigma)+l(\tau)-1.
\end{align}
Hence,
  \begin{align}
    &\sum_{k_{i_1},\ldots,k_{i_M}=1}^d\sum_{l_{j_1},\ldots,l_{j_M}=1}^d\delta_{l_{j_1},k_{\sigma^{-1}(i_1)}}\delta_{k_{i_1},l_{\tau^{-1}(j_1)}}\nonumber
    \\&\times\prod_{p=2}^M\delta_{k_{i_p},k_{\sigma^{-1}(i_p)}}\delta_{l_{j_p},l_{\tau^{-1}(j_p)}}\nonumber
    \\=&\ d^{l(\sigma)+l(\tau)-1}.
  \end{align}
Here, the number of $\sigma\in S_M$ satisfying $l(\sigma)=k$ is given by the first kind of Stirling number $\sterling{M}{k}$ (see Remark~\ref{rem:Stirling numbers} for details). Therefore,
\begin{align}
  \sum_{\sigma\in S_I}\sum_{\tau\in S_J}d^{l(\sigma)+l(\tau)-1}&=\left(\sum_{\sigma\in S_I}d^{l(\sigma)}\right)\left(\sum_{\tau\in S_J}d^{l(\tau)}\right)d^{-1}\nonumber
  \\&=\left(\sum_{k=1}^M\sterling{M}{k}d^k\right)^2d^{-1}\nonumber
  \\&=\frac{1}{d}\left(\frac{(M+d-1)!}{(d-1)!}\right)^2.\label{eq:Stirling}
\end{align}
Thus,
\begin{align}
  \Tr\left[\eta^I_{A^NB}\eta^J_{A^NB}\right]&=\frac{1}{(d[M]M!)^2d^N}\frac{1}{d}\left(\frac{(M+d-1)!}{(d-1)!}\right)^2\nonumber
  \\&=\frac{1}{d^{N+1}}.
\end{align}
\end{proof}

\begin{rem}\label{rem:Stirling numbers}
  The first kind of Stirling number $\sterling{n}{k}$ is defined as the coefficient of $x^k$ in the expansion of the rising factorial
  \begin{align}
    x^{\bar{n}}\coloneq x(x+1)\cdots(x+n-1)
  \end{align}
  as a power series in $x$:
  \begin{align}
    x^{\bar{n}}=\sum_{k=0}^n\sterling{n}{k}x^k.
  \end{align}
  It is known that $\sterling{n}{k}$ gives the number of ways to decompose a set of $n$ elements into $k$ cycles. The following relationship was used in \eqref{eq:Stirling}:
  \begin{align}
    \sum_{k=0}^M\sterling{M}{k}d^k&=d^{\bar{M}}=d(d+1)\cdots(d+M-1)\nonumber
    \\&=\frac{(M+d-1)!}{(d-1)!}.
  \end{align}
\end{rem}

\begin{prop}\label{prop:asymptotic value of bar eta}
  For $\bar{\eta}_{A^NB}={\binom{N}{M}}^{-1}\sum_{I\in\mathcal{I}_N^M}\eta^I_{A^NB}$, the following holds:
  \begin{align}\label{eq:asymptotic value of bar eta}
      \lim_{N\rightarrow\infty}d^{N+1}\Tr\left[\bar{\eta}_{A^NB}^2\right]=1.
  \end{align}
\end{prop}

\begin{proof}
Let $m(k)\ (0\leq k\leq M)$ be the number of pairs $(I,J)\ (I,J\in\mathcal{I}_N^M)$ that satisfy $|I\cap J|=k$, and let $\nobreak{f(k)=\Tr\big[\eta^I_{A^NB}\eta^J_{A^NB}\big]}$ when $|I\cap J|=k$. Note from Lemma \ref{lem:Tr depends on multiplicity of alphabets} that $f(k)$ depends only on $k$. Then, 
\begin{align}
&d^{N+1}\Tr\left[\bar{\eta}_{A^NB}^2\right]\nonumber
\\=&d^{N+1}\binom{N}{M}^{-2}\sum_{I,J\in\mathcal{I}_N^M}\Tr\left[\eta^I_{A^NB}\eta^J_{A^NB}\right]\nonumber
\\=&d^{N+1}\binom{N}{M}^{-2}\sum_{k=0}^Mm(k)f(k)\nonumber
\\\leq&d^{N+1}\binom{N}{M}^{-2}\left(m(0)f(0)+\left(\sum_{k=1}^Mm(k)\right)f(M)\right)\nonumber
\\=&d^{N+1}\left(\binom{N-M}{M}\binom{N}{M}^{-1}f(0)\right.\nonumber
\\&\left.+\left(1-\binom{N-M}{M}\binom{N}{M}^{-1}\right)f(M)\right)\nonumber
\\\leq& d^{N+1}\left(\binom{N-M}{M}\binom{N}{M}^{-1}\frac{1}{d^{N+1}}\right.\nonumber
\\&\left.+\left(1-\binom{N-M}{M}\binom{N}{M}^{-1}\right)\frac{d^{M-N+2}}{d[M]}\frac{M!(M-1)!}{d+M-1}\right)\nonumber
\\=&\binom{N-M}{M}\binom{N}{M}^{-1}+\left(1-\binom{N-M}{M}\binom{N}{M}^{-1}\right)\nonumber
\\&\times\frac{d^{M+3}}{d[M]}\frac{M!(M-1)!}{d+M-1}.
\end{align}
The first inequality uses Lemma \ref{lem:Tr takes maximum value when I=J}, and the second inequality uses Lemmas \ref{lem:upper bound of Tr when I=J} and \ref{lem:value of Tr when I cap J = 0}. Since
\begin{align}
  \lim_{N\rightarrow\infty}\binom{N-M}{M}\binom{N}{M}^{-1}=1
\end{align}
holds for finite $M$, we obtain
\begin{align}
    \lim_{N\rightarrow\infty}\left|d^{N+1}\Tr\left[\bar{\eta}_{A^NB}^2\right]-1\right|=0.
\end{align}
\end{proof}

Finally, we prove Theorem \ref{thm:main result}.
\begin{proof}[Proof of Theorem \ref{thm:main result}]
Let $\{E^I_{XA^N}\}_{I\in\mathcal{I}_N^M}$ be the PGM for $\{(|\mathcal{I}_N^M|^{-1},\eta^I_{A^NB})\}_{I\in\mathcal{I}_N^M}$. From Proposition \ref{prop:P invariance of E}, the following holds:
\begin{align}
    p_\mathrm{succ}\coloneq&{\binom{N}{M}}^{-1}\sum_{I\in\mathcal{I}_N^M}\Tr\left[E^{I}_{A^NB}\eta^{I}_{A^NB}\right]\nonumber
    \\=&{\binom{N}{M}}^{-1}\frac{d^M}{d[M]}\sum_{I\in\mathcal{I}_N^M}\Tr\left[E^{I}_{A^NB}\rho^{i_1}_{A^NB}\right].
\end{align}
Thus, we obtain
\begin{align}
    \sum_{I\in\mathcal{I}_N^M}\Tr\left[E^{I}_{A^NB}\rho^{i_1}_{A^NB}\right]=\frac{ d[M]}{d^M}\binom{N}{M}p_\mathrm{succ}.
\end{align}
Therefore,
\begin{align}
    &F(\mathcal{R}\circ\mathcal{D}^\mathrm{std}_{N,M})\nonumber
    \\=&\frac{1}{d^2}\sum_{I\in\mathcal{I}_N^M}\Tr\left[E^I_{A^NB}\rho^{i_1}_{A^NB}\right]\nonumber
    \\=&\frac{1}{d^2}\frac{d[M]}{d^M}\binom{N}{M}p_\mathrm{succ}\nonumber
    \\\geq&\frac{d[M]}{d^{M+2}}\binom{N}{M}\binom{N}{M}^{-1}\frac{1}{d[M-1]d^{N-M}}\frac{1}{\Tr\left[\bar{\eta}_{A^NB}^2\right]}\nonumber
    \\=&\frac{d+M-1}{dM}\frac{1}{d^{N+1}\Tr\left[\bar{\eta}_{A^NB}^2\right]}.
\end{align}
The first equality uses Corollary \ref{cor:F of PBTC}, and the first inequality uses Lemma \ref{lem:general lower bound of PGM} and Proposition \ref{prop:bar r}. From Proposition \ref{prop:asymptotic value of bar eta},
\begin{align}
    \lim_{N\rightarrow\infty}F(\mathcal{R}\circ\mathcal{D}^\mathrm{std}_{N,M})\geq\frac{d+M-1}{dM}.
\end{align}
Thus, from \eqref{eq:f and F},
\begin{align}
    \lim_{N\rightarrow\infty}f(\mathcal{R}\circ\mathcal{D}^\mathrm{std}_{N,M})\geq\frac{d+2M-1}{M(d+1)}.
\end{align}
On the other hand, since the fidelity of symmetric cloning is upper bounded by \eqref{eq:optimal fidelity}, the equality holds.
\end{proof}

\section{Conclusion}\label{sec:Conclusion}
In this paper, we introduced port-based telecloning, a variant of telecloning that uses PBT instead of conventional teleportation.
To achieve this, we constructed a new POVM by partially symmetrizing the state that constitutes the optimal POVM for PBT. We then demonstrated that the PBTC protocol we constructed can asymptotically distribute optimal clones to many receivers. Furthermore, numerical calculations showed that, at least in the case of few ports, PBTC outperforms the naive clone-and-teleport protocol.

There are several open questions about PBTC. The first is finding an optimal POVM for PBTC with finite $N$. We showed that the POVM we introduced achieves an optimal value in the limit  $N\rightarrow\infty$, but its optimality for finite $N$ has not been clarified. In previous research~\cite{studzinski2017port,mozrzymas2018optimal,leditzky2022optimality}, the optimal POVM in PBT was derived using semidefinite programming and representation theory. Since PBTC additionally requires the condition to be a symmetric cloning, the proof done in PBT cannot be directly applied to PBTC, but it is expected that a similar method can be used. In addition, since the results of this study were obtained for the maximally entangled resource states, the optimization of a resource state can also be considered. It could increase the efficiency of PBTC since in the optimized PBT we have a square improvement of fidelity in $N$. Also, we considered only the deterministic PBTC, but the study of a probabilistic version presents an additional challenge.
Regarding the numerical study, it would be interesting to verify whether the performance gap between our PBTC and the naive version persists for larger numbers of ports.

\section*{Acknowledgement}
We thank S. Strelchuk for helpful comments. K.K. acknowledges support from JSPS Grant-in-Aid for Early-Career Scientists No. 22K13972 and from MEXT-JSPS Grant-in-Aid for Transformative Research Areas (B) No. 24H00829. F.B. acknowledges support from MEXT Quantum Leap Flagship Program (MEXT QLEAP) Grant No.~JPMXS0120319794, from MEXT-JSPS Grant-in-Aid for Transformative Research Areas (A) ``Extreme Universe'' Grant No.~21H05183, and from JSPS KAKENHI Grant No.~23K03230.

\bibliography{main}

\clearpage
\onecolumngrid
\appendix
\section{Detailed derivation of Equation \eqref{eq:appendix}}\label{appendix}
In this appendix, we give a detailed derivation of Eq. \eqref{eq:appendix}.
At first, for $I\in\mathcal{I}_N^M=\{i_1,\ldots,i_M\}$,
\begin{align}
    \Pi_{A_I}\rho^{i_1}_{A^NB}&=\frac{1}{M!d^N}\sum_{\sigma\in S_I}\sum_{k_{i_1},\ldots,k_{i_M},k'_{i_1}=1}^d V_\sigma\ket{k_{i_1}k_{i_2}\cdots k_{i_M}}\bra{k'_{i_1}k_{i_2}\cdots k_{i_M}}_{A_I}\otimes\ket{k_{i_1}}\bra{k'_{i_1}}_B\otimes\mathbbm{1}_{A^N\backslash A
    _I}\nonumber
    \\&=\frac{1}{M!d^N}\sum_{\sigma\in S_I}\sum_{k_{i_1},\ldots,k_{i_M},k'_{i_1}=1}^d \ket{k_{\sigma^{-1}(i_1)}k_{\sigma^{-1}(i_2)}\cdots k_{\sigma^{-1}(i_M)}}\bra{k'_{i_1}k_{i_2}\cdots k_{i_M}}_{A_I}\otimes\ket{k_{i_1}}\bra{k'_{i_1}}_B\otimes\mathbbm{1}_{A^N\backslash A_I}.
\end{align}
Likewise, for $J=\{j_1,\ldots,j_M\}\in\mathcal{I}_N^M$,
\begin{align}
    \Pi_{A_J}\rho^{j_1}_{A^NB}=\frac{1}{M!d^N}\sum_{\tau\in S_J}\sum_{l_{j_1},\ldots,l_{j_M},l'_{j_1}=1}^d \ket{l_{\tau^{-1}(j_1)}l_{\tau^{-1}(j_2)}\cdots l_{\tau^{-1}(j_M)}}\bra{l'_{j_1}l_{j_2}\cdots l_{j_M}}_{A_J}\otimes\ket{l_{j_1}}\bra{l'_{j_1}}_B\otimes\mathbbm{1}_{A^N\backslash A_J}.
\end{align}
Hence, when $I\cap J=\emptyset$,
\begin{align}
    &\Tr\left[\Pi_{A_I}\rho^{i_1}_{A^NB}\Pi_{A_J}\rho^{j_1}_{A^NB}\right]\nonumber
    \\=&\frac{1}{(M!)^2d^{2N}}\sum_{\sigma\in S_I}\sum_{\tau\in S_J} \sum_{k_{i_1},\ldots,k_{i_M},k'_{i_1}=1}^d\sum_{l_{j_1},\ldots,l_{j_M},l'_{j_1}=1}^d \Tr\Big[\ket{k_{\sigma^{-1}(i_1)}k_{\sigma^{-1}(i_2)}\cdots k_{\sigma^{-1}(i_M)}}\bra{k'_{i_1}k_{i_2}\cdots k_{i_M}}_{A_I}\nonumber
    \\&\otimes\ket{l_{\tau^{-1}(j_1)}l_{\tau^{-1}(j_2)}\cdots l_{\tau^{-1}(j_M)}}\bra{l'_{j_1}l_{j_2}\cdots l_{j_M}}_{A_J}\otimes\ket{k_{i_1}}\braket{k'_{i_1}|l_{j_1}}\bra{l'_{j_1}}_B\otimes\mathbbm{1}_{A^N\backslash A_IA_J}\Big]\nonumber
    \\=&\frac{1}{(M!)^2d^{N+2M}}\sum_{\sigma\in S_I}\sum_{\tau\in S_J} \sum_{k_{i_1},\ldots,k_{i_M},k'_{i_1}=1}^d\sum_{l_{j_1},\ldots,l_{j_M},l'_{j_1}=1}^d \nonumber
    \\&\times\delta_{k'_{i_1},k_{\sigma^{-1}(i_1)}}\delta_{k_{i_2},k_{\sigma^{-1}(i_2)}}\cdots\delta_{k_{i_M},k_{\sigma^{-1}(i_M)}}\ \delta_{l'_{j_1},l_{\tau^{-1}(j_1)}}\delta_{l_{j_2},l_{\tau^{-1}(j_2)}}\cdots\delta_{l_{j_M},l_{\tau^{-1}(j_M)}}\ \delta_{k_{i_1},l'_{j_1}}\ \delta_{k'_{i_1},l_{j_1}}\nonumber
    \\=&\frac{1}{(M!)^2d^{N+2M}}\sum_{\sigma\in S_I}\sum_{\tau\in S_J} \sum_{k_{i_1},\ldots,k_{i_M}=1}^d\sum_{l_{j_1},\ldots,l_{j_M}=1}^d \delta_{l_{j_1},k_{\sigma^{-1}(i_1)}}\delta_{k_{i_1},l_{\tau^{-1}(j_1)}}\prod_{p=2}^M\delta_{k_{i_p},k_{\sigma^{-1}(i_p)}}\delta_{l_{j_p},l_{\tau^{-1}(j_p)}}.
\end{align}
Therefore, from \eqref{eq:Tr}, we obtain \eqref{eq:appendix}.

\end{document}